%% file: mainfile.tex
\documentclass[conference]{IEEEtran}
\IEEEoverridecommandlockouts

\usepackage[skip=0pt]{caption}
\usepackage[skip=0pt]{subcaption}

\usepackage{times}
\usepackage{latexsym}

\usepackage[T1]{fontenc}
\usepackage[utf8]{inputenc}
\usepackage{microtype}
\usepackage{inconsolata}
\usepackage{graphicx}

\newcommand{\myparatight}[1]{\smallskip\noindent{\bf {#1}:}~}

\usepackage{amsthm}
\usepackage{amsmath}
     
\usepackage{amssymb}
\usepackage{tcolorbox}
\usepackage{xcolor}

\usepackage{booktabs}
\usepackage{array} 
\usepackage{multirow}
\usepackage{makecell}
\usepackage{graphics}
\usepackage{graphicx}

\usepackage{algorithm}
\usepackage{algpseudocode}
\usepackage{threeparttable}

\usepackage{amsfonts}
\usepackage{color, url}

\usepackage{makecell}

\usepackage{tabto}
\usepackage{xcolor, colortbl}
\usepackage{enumitem}
\usepackage{bbm}
\usepackage{dsfont}

\usepackage{bm}
\usepackage[mathscr]{eucal}

\usepackage{bigstrut,multirow,rotating}
\usepackage{placeins}
\usepackage{pgfplots}
\pgfplotsset{compat=1.17}
\makeatletter
\@namedef{r@tocindent4}{0pt}
\makeatother

\newtheorem{thm}{Theorem}

\usepackage{xspace}

\newcommand{\alg}{{RAGuard}\xspace}

\definecolor{myurlcolor}{rgb}{0.1, 0.2, 0.8}

\usepackage{hyperref}
\hypersetup{
	colorlinks=true,
	linkcolor=blue,
	filecolor=magenta,      
	urlcolor=myurlcolor,
	citecolor=blue,
	pdftitle={Overleaf Example},
	pdfpagemode=FullScreen,
}
\usepackage{cite}

\usepackage{color, colortbl}

\definecolor{greyBackground}{RGB}{230,248,255}

\newcolumntype{a}{>{\columncolor{Gray}}c}
\newcolumntype{b}{>{\columncolor{greyL}}c}

\algnewcommand\algorithmicforpara{\textbf{for}}
\algnewcommand\algorithmicdoinparallel{\textbf{do in parallel}}
\algdef{S}[FOR]{ForParallel}[1]{\algorithmicforpara\ #1\ \algorithmicdoinparallel}

\def\BibTeX{{\rm B\kern-.05em{\sc i\kern-.025em b}\kern-.08em
		T\kern-.1667em\lower.7ex\hbox{E}\kern-.125emX}}

\begin{document}

\title{Secure Retrieval-Augmented Generation against Poisoning Attacks

}

\author{\IEEEauthorblockN{
Zirui Cheng\IEEEauthorrefmark{2}\textsuperscript{*},
Jikai Sun\IEEEauthorrefmark{2}\textsuperscript{*},
Anjun Gao\IEEEauthorrefmark{5},
Yueyang Quan\IEEEauthorrefmark{3},
Zhuqing Liu\IEEEauthorrefmark{3}, 
Xiaohua Hu\IEEEauthorrefmark{4},
Minghong Fang\IEEEauthorrefmark{5}}
\IEEEauthorblockA{
\IEEEauthorrefmark{2}National University of Singapore, 
\IEEEauthorrefmark{5}University of Louisville, 
\IEEEauthorrefmark{3}University of North Texas, \IEEEauthorrefmark{4}Drexel University}
}

\maketitle

\makeatletter
\renewcommand{\@makefntext}[1]{\noindent#1}
\makeatother
\footnotetext{\textsuperscript{*}Equal contribution. Zirui Cheng and Jikai Sun conducted this research while they were interns under the supervision of Minghong Fang.}

\begin{abstract}	 
Large language models (LLMs) have transformed natural language processing (NLP), enabling applications from content generation to decision support. Retrieval-Augmented Generation (RAG) improves LLMs by incorporating external knowledge but also introduces security risks, particularly from data poisoning, where the attacker injects poisoned texts into the knowledge database to manipulate system outputs. While various defenses have been proposed, they often struggle against advanced attacks. To address this, we introduce RAGuard, a detection framework designed to identify poisoned texts. RAGuard first expands the retrieval scope to increase the proportion of clean texts, reducing the likelihood of retrieving poisoned content. It then applies chunk-wise perplexity filtering to detect abnormal variations and text similarity filtering to flag highly similar texts. This non-parametric approach enhances RAG security, and experiments on large-scale datasets demonstrate its effectiveness in detecting and mitigating poisoning attacks, including strong adaptive attacks.
\end{abstract}

\begin{IEEEkeywords}
	Retrieval-Augmented Generation, Poisoning Attacks, Robustness
\end{IEEEkeywords}

\input{introduction}

\input{related}

\input{threatModel}

\input{method}

\input{experiments}

\input{discussion}

\input{conclusion}

\bibliography{refs}
\bibliographystyle{IEEEtran}

\input{appendix}

\end{document}

%% file: introduction.tex

\section{Introduction} \label{sec:intro}

Large language models (LLMs)~\cite{brown2020language} have advanced natural language processing (NLP).
However, they often generate false or misleading information due to outdated or limited training data, which can be critical in domains like finance and healthcare. The Retrieval-Augmented Generation (RAG) framework~\cite{karpukhin2020dense} mitigates this issue by enriching LLMs with external knowledge. Upon receiving a query, RAG retrieves semantically relevant documents from a database and uses them to ground the LLM’s response, improving factual accuracy and contextual relevance.

Including external data sources introduces security vulnerabilities, leaving RAG systems exposed to poisoning attacks~\cite{zou2024poisonedrag}. Attackers can inject crafted malicious texts into the knowledge base, causing the system to generate attacker-specified responses for certain queries. Existing defenses fall into two main types: prevention and detection.
Prevention-based defenses~\cite{jain2023baseline,xiang2024certifiably} modify how queries are processed to reduce attack opportunities, such as using query paraphrasing to avoid triggering poisoned responses. Detection-based defenses~\cite{zhou2025trustrag} seek to identify and remove malicious texts from the knowledge base.
However, as shown in our experiments, these defenses often fail to effectively counter advanced poisoning attacks.

\myparatight{Our work}%
In this paper, we introduce \alg, a novel approach designed to effectively identify poisoned texts within RAG systems. Poisoning attacks in RAG involve the attacker inserting carefully crafted poisoned texts into the knowledge database, aiming to manipulate the system to produce specific attacker-controlled responses for particular queries. For such attacks to succeed, it is crucial that the poisoned texts appear in the retrieval phase, as the retrieval results directly influence the model’s generated responses.
Motivated by this critical requirement, we propose expanding the number of retrieved texts per query as a proactive defense strategy. This retrieval expansion approach is based on the insight that a larger retrieval set naturally includes a higher proportion of legitimate, benign texts, which significantly reduces the probability of poisoned texts dominating the retrieved results. By diluting the concentration of poisoned texts among the retrieved texts, our strategy effectively diminishes their potential influence on the RAG system’s output.

However, retrieval expansion alone may not fully eliminate the risk posed by poisoned content, especially if the poisoned texts are deliberately crafted to closely match query semantics. Therefore, we further propose a two-stage filtering strategy to robustly identify and remove poisoned texts from the expanded retrieval results. The motivation behind employing a two-stage filtering mechanism is derived from our empirical observation that poisoned texts often exhibit distinct linguistic irregularities compared to benign texts. In particular, poisoned texts typically show substantial discrepancies in perplexity scores between their two constituent chunks or display unusually high perplexity in at least one chunk, indicating unnatural linguistic patterns.
Furthermore, the attacker frequently attempts to maximize the likelihood of their poisoned texts being retrieved by embedding phrases that closely mirror the query content. To counteract this strategy, we incorporate an additional filtering step that specifically targets texts exhibiting abnormally high semantic similarity to the query. By systematically removing these suspiciously similar texts, our two-stage filtering strategy further enhances the accuracy and reliability of poisoned content detection.
Together, the expansion of the retrieval pool and the two-stage filtering approach constitute a cohesive defense framework. Our comprehensive approach integrates complementary strategies, including dilution of the proportion of poisoned texts and rigorous detection of subtle linguistic anomalies, to effectively safeguard RAG systems against poisoning attacks.

We comprehensively evaluate our method across five datasets and 6 poisoning attacks, including 2 advanced adaptive attacks where the attacker evades detection through stealthy strategies. We benchmark our approach against 6 established defense baselines and explore practical RAG scenarios, such as using different LLMs as evaluation backbones and applying diverse text similarity metrics. The main contributions are as follows:

\begin{list}{\labelitemi}{\leftmargin=1em \itemindent=-0.00em \itemsep=.2em}

\item We introduce \alg, an innovative defense strategy designed to counteract poisoning attacks on RAG systems.

\item We conduct a thorough empirical evaluation of \alg against a range of existing attacks, demonstrating that \alg is highly effective in identifying and detecting poisoned texts.

\item We develop tailored adaptive attacks targeting \alg and assess their impact. Our findings reveal that \alg maintains strong resilience, effectively defending against these adaptive attacks.

\end{list}

%% file: related.tex

\section{Preliminaries and Related Work} \label{sec:related}

\subsection{Retrieval-Augmented Generation (RAG)}

RAG~\cite{karpukhin2020dense} combines retrieval with generative models to enhance knowledge-intensive NLP tasks. By accessing external knowledge, it improves understanding and response generation. A typical RAG system includes three main components: {\em a knowledge database}, {\em a retrieval module}, and {\em a large language model (LLM)}. For a user query $Q$, the RAG system generates a response in two phases.

\begin{list}{\labelitemi}{\leftmargin=1em \itemindent=-0.0em \itemsep=.2em}

\item \textbf{Phase I (Knowledge retrieval):} 
When a query $Q$ is received, the retriever selects the top-$k$ texts from the knowledge database $D$ by computing semantic similarity between $Q$ and all texts in $D$. The $k$ most similar texts are retrieved, denoted as $H(Q, D, k)$.

\item \textbf{Phase II (Response generation):} 
After retrieving the top-$k$ relevant texts, the LLM uses a predefined prompt to generate the response based on the user's query and retrieved texts. 
\end{list}

Fig.~\ref{RAG_fig} shows the process by which a RAG system generates a response for a user query.

\begin{figure}[!t]
	\centering
   	\includegraphics[scale = 0.25]{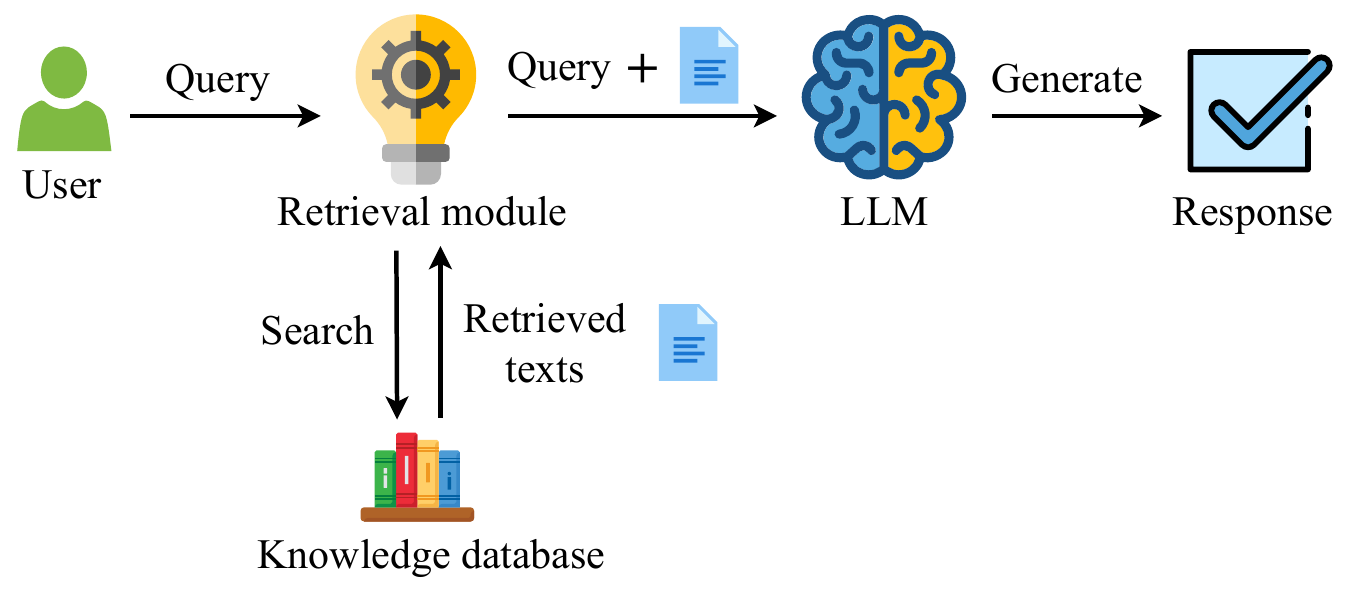}
	\caption{Illustration of the RAG process.}
        \label{RAG_fig}
		 \vspace{-.2in}
\end{figure}

\subsection{Poisoning Attacks and Defenses to RAG}

Although RAG systems enhance language models with external knowledge, they remain vulnerable to knowledge poisoning, where corrupted databases can manipulate outputs~\cite{zou2024poisonedrag,zhang2025practical,zhang2025benchmarking}.
Traditional defenses~\cite{barreno2010security,cretu2008casting,suciu2018does} are ineffective since RAG attacks target retrieval instead of training. Recent RAG-specific defenses~\cite{jelinek1980interpolated,jain2023baseline,xiang2024certifiably,alon2023detecting,liu2024formalizing,zhou2025trustrag} fall into prevention- and detection-based categories. Prevention approaches like query paraphrasing~\cite{jain2023baseline} or isolate-then-aggregate mechanisms~\cite{xiang2024certifiably} mitigate risk but do not identify poisoned texts, while detection-based methods~\cite{zhou2025trustrag} often miss subtle manipulations, leaving RAG systems vulnerable.
Note that we do not consider post-attack defenses such as forensic techniques~\cite{zhang2025traceback,zhang2025taught}.

%% file: threatModel.tex

\section{Problem Formulation} 

\subsection{Threat model} 

We adopt the threat model introduced in prior studies~\cite{zou2024poisonedrag}, where the attacker’s objective is to compromise the RAG system by injecting multiple poisoned texts into the knowledge database. These texts are carefully crafted so that, when a target query \( Q \) is issued, the system retrieves the poisoned content and generates a response aligned with the attacker’s intent.
A RAG system consists of three core components: a knowledge database, a retrieval module, and an LLM. We assume a realistic attack model where the attacker lacks access to the knowledge database's existing content, cannot query the LLM, and is unaware of its parameters, aligning with~\cite{zou2024poisonedrag}.
Based on the attacker’s knowledge of the retrieval module, attacks are categorized as white-box or black-box. In the white-box setting, the attacker fully knows the retrieval parameters, while in the black-box setting, they do not. Since the white-box scenario represents a stronger adversary, we focus on it to demonstrate our defense’s robustness against even the most advanced attacks.

\subsection{Defense objectives}

Our goal is to develop a reliable detection method to protect RAG systems from poisoning attacks.
The detection method should ensure \textit{preservation of system integrity}, maintaining the RAG system’s performance in benign conditions by avoiding false flags on harmless texts that could otherwise affect response quality and functionality. It must also achieve \textit{detection precision}, accurately identifying poisoned texts within the knowledge database and effectively distinguishing them from benign ones with minimal errors, which is particularly critical when attackers attempt to disguise poisoned content as benign. Furthermore, the method should guarantee \textit{computational efficiency}, introducing minimal computational overhead so that resource demands remain comparable to a standard RAG system.
Notably, the defender lacks prior information about the attacker, including their methods for crafting poisoned texts or the exact number of poisoned texts in the knowledge database.

%% file: method.tex

\section{Our \alg} 
\label{our_method}

\subsection{Overview}

In poisoning attacks on RAG systems, the attacker injects carefully crafted poisoned texts into the knowledge database to influence the system's response to specific queries. For the attack to work, these poisoned texts must be among the texts retrieved in Phase I of the RAG process. To counter this, we propose a strategy that increases the proportion of clean texts in the retrieved set, thereby reducing the chance of retrieving poisoned content. 
Our \alg method achieves this by expanding the retrieval scope to include more than the typical top-$k$ texts, retrieving the top-$N$ texts instead, where $N > k$. This increases the likelihood of retrieving benign texts and reduces the impact of attacks.

Once the retrieval set is expanded, we introduce a two-stage filtering strategy aimed at detecting and removing poisoned texts from such enlarged set of candidates. Let \(\hat{D}\) denote the knowledge database after the injection of multiple poisoned texts into the original clean database \(D\). For a given query \(Q\), we use \(H(Q, \hat{D}, N)\) to represent the top-\(N\) retrieved texts from \(\hat{D}\). Our goal is to determine, for each retrieved text \(R \in H(Q, \hat{D}, N)\), whether it is benign or has been manipulated by the attacker. The following sections describe in detail how our two-stage filtering mechanism systematically assesses and filters these retrieved texts.

\subsection{Chunk-wise perplexity filtering} 
\label{Chunk_wise_filtering}

A straightforward way to determine whether a text \( R  \) in $H(Q, \hat{D}, N)$ is poisoned is by evaluating its perplexity score, which measures text quality. A lower perplexity score indicates higher quality and a greater likelihood of the text being benign, while a higher perplexity score suggests lower quality and a higher probability of being poisoned. The perplexity score for a given text \( R \) is computed as follows:
\begin{align}
f(R) = -\frac{1}{|R|}\sum_{r \in R}\log p(r_i|r_{0:i-1}),
\end{align}
where $R \in H(Q, \hat{D}, N)$, \( |R| \) represents the length of the text, i.e., the total number of words in \( R \). The term \( r_i \) denotes the word at the \( i \)-th position in the text, while \( r_{0:i-1} \) refers to the sequence of words from the beginning of the text up to, but not including, the \( i \)-th word. Lastly, \( p(r_i \mid r_{0:i-1}) \) is the probability of the word \( r_i \) occurring, given the preceding sequence \( r_{0:i-1} \), as estimated by the language model.

However, as our experiments later demonstrate, computing the perplexity score for an entire text~\cite{liu2024formalizing,jelinek1980interpolated} proves ineffective. To overcome this limitation, we introduce a chunk-based approach. 
This approach is motivated by the observation that poisoned texts often exhibit significant fluctuations in perplexity across different sections due to their lack of fluency and coherence, whereas benign texts maintain more consistent perplexity scores. 
To capture these variations, we divide each text in the set \(H(Q, \hat{D}, N)\) into two chunks, either by approximately splitting the text evenly or using punctuation marks~\cite{izacard2020leveraging,khattab2020colbert}. 
Our experiments confirm that a two-chunk split is sufficient.
For a given text \(R\), we denote the first and second chunks as \(R^\text{pre}\) and \(R^\text{post}\), respectively, such that \(R=R^\text{pre} \oplus R^\text{post}\).

This chunk-based approach enables us to capture variations in perplexity within a text, which can be indicative of poisoning. To quantify this variation, we introduce the Perplexity Difference (PD) score, which measures the difference in perplexity between the two chunks. Formally, it is defined as:
\begin{align}
\label{pd_score}
\text{PD}(R) = f(R^\text{pre}) - f(R^\text{post}).
\end{align}  

Since \( \text{PD}(R) \) can exhibit both high and low values depending on the text structure, we employ a non-parametric detection mechanism that utilizes the empirical distribution of PD scores rather than relying on parametric assumptions. Specifically, we randomly sample a subset of texts \( \mathcal{S} \) from the knowledge database and compute \( \text{PD}(\upsilon) \) for each \( \upsilon \in \mathcal{S} \). Using these scores, we determine the empirical percentile thresholds to establish a rejection region for hypothesis testing. Specifically, we formulate the
following two hypotheses:
\begin{list}{\labelitemi}{\leftmargin=1em \itemindent=-0.0em \itemsep=.2em}
\item \textbf{Null hypothesis \( H_0 \)}: The text \( R \) is not poisoned, and its PD score \( \text{PD}(R) \) falls within the range observed in \( \mathcal{S} \).

\item \textbf{Alternative hypothesis \( H_1 \)}: The text \( R \) is poisoned, and its PD score significantly deviates from the PD scores of texts in \( \mathcal{S} \).

\end{list}

Let \( Q_{\alpha}^{\text{PD}} \) and \( Q_{1-\alpha}^{\text{PD}} \) denote the \(\alpha\)-th and \((1-\alpha)\)-th percentiles of the PD scores in \( \mathcal{S} \), respectively. A text is classified as poisoned if its PD score falls within this region:  
\begin{align}
\label{pd_condition}
  \!\!\!  \text{PD}(R) \geq Q_{1-\alpha}^{\text{PD}} \quad \text{or} \quad \text{PD}(R) \leq Q_{\alpha}^{\text{PD}}.
\end{align}

If Eq.~(\ref{pd_condition}) holds, we reject the null hypothesis \( H_0 \), and classify \( R \) as poisoned. 
The significance level \( \alpha \) controls the sensitivity of the threshold. For instance, setting \( \alpha = 2.5\% \) ensures that at most 2.5\% of benign texts fall outside the empirical percentile range due to random variation. This distribution-free approach aligns with the empirical distribution of PD scores, eliminating the need for specific assumptions, such as a Gaussian distribution, while preserving effective detection performance.

However, relying solely on the PD score may not always effectively detect a poisoned text \( R \), as both \( f(R^\text{pre}) \) and \( f(R^\text{post}) \) can be large, potentially masking the poisoning by producing a lower PD score.  
To address this limitation, we introduce an additional metric called the Perplexity Maximum (PM) score, which captures the maximum perplexity between the two chunks of a text \( R \). Formally, the PM score \( \text{PM}(R) \) is defined as:  
\begin{align}
\label{pm_score}
\text{PM}(R) = \max(f(R^\text{pre}), f(R^\text{post})).
\end{align}  

A high PM score suggests that at least one chunk of \( R \) is of significantly lower quality. To determine whether \( \text{PM}(R) \) indicates poisoning, we again apply a hypothesis testing approach based on empirical percentile thresholds.  
Specifically, for each text $v$ in the selected set \( \mathcal{S} \), we compute \( \text{PM}(\upsilon) \). Instead of fitting a parametric model, we estimate the upper empirical percentile of the PM scores in \( \mathcal{S} \) to define the rejection threshold. Let \( Q_{1-\alpha}^{\text{PM}} \) denote the \((1-\alpha)\)-th percentile of the PM scores in \( \mathcal{S} \), then a text \( R \) is flagged as poisoned if:  
\begin{align}
\label{pm_condition}
\text{PM}(R) \geq Q_{1-\alpha}^{\text{PM}}.
\end{align}  

Note that since a high PM score indicates that at least one chunk of \( R \) exhibits unusually high perplexity, we thus apply only a one-tailed test. Low PM scores are not a concern, as they suggest both chunks fall within the normal perplexity range, characteristic of benign texts.

\subsection{Text similarity filtering}

While chunk-wise perplexity filtering helps ensure that the top-\(N\) retrieved texts are coherent, it may still miss poisoned texts designed to appear natural. The attacker can craft such texts to evade detection based solely on perplexity.  %
To address this, we introduce a text similarity-based filter. In poisoning attacks, injected texts must closely resemble the query \( Q \) to increase their retrieval chances. Techniques like the HotFlip algorithm~\cite{ebrahimi2017hotflip} are commonly used to achieve this effect. As a result, poisoned texts tend to exhibit {\em unusually high similarity} to \( Q \) compared to benign texts, which typically have moderate similarity scores.  
Specifically, 
for a given query \( Q \), we compute its embedding \( E(Q) \) and calculate similarity scores for each retrieved text \( R \) using:  
\begin{align}
\label{ts_score}
    \text{TS}(R) = \text{Sim}(E(Q), E(R)),
\end{align}  
where \( \text{Sim}() \) is a similarity function (e.g., cosine similarity). 

Let \( Q_{1-\alpha}^{\text{TS}} \) denote the \((1-\alpha)\)-th percentile of text similarity scores in \( \mathcal{S} \), then a text \( R \) is classified as poisoned if:  
\begin{align}
\label{ts_condition}
\text{TS}(R) \geq Q_{1-\alpha}^{\text{TS}}.
\end{align}  

This one-tailed test focuses only on high similarity scores, as the retrieval system selects texts with the highest similarity in Phase I. By adding this similarity-based filter, we improve the detection of poisoned texts that might pass undetected by perplexity alone, enhancing the RAG system’s robustness against attacks.

\begin{algorithm}[t]
	\caption{\alg.}
	\label{our_alg}
	\begin{algorithmic}[1]
		\renewcommand{\algorithmicrequire}{\textbf{Input:}}
		\renewcommand{\algorithmicensure}{\textbf{Output:}}
		\Require Knowledge database $\hat{D}$, query $Q$; parameters $k$, $N$.
		\Ensure Final response to query \( Q \). 
		
		\State  Select a random subset of texts, denoted as $\mathcal{S}$, from  $\hat{D}$.

		\State  Retrieve the top-\( N \) texts most relevant to query \( Q \) based on Phase I of the RAG system, forming the set \( H(Q, \hat{D}, N) \).
		
		\For{each text $R$ in \( H(Q, \hat{D}, N) \)}
		
		\State  Split $R$ into two chunks, denoted as \(R^\text{pre}\) and \(R^\text{post}\).
		
		\State Calculate $\text{PD}(R)$, $\text{PM}(R)$, and $\text{TS}(R)$ using Eq.~(\ref{pd_score}), Eq.~(\ref{pm_score}), and Eq.~(\ref{ts_score}), respectively.
		\If {any of the criteria specified in Eq.~(\ref{pd_condition}), Eq.~(\ref{pm_condition}), or Eq.~(\ref{ts_condition}) are satisfied}
		
		\State Classify text $R$ as poisoned, and remove $R$ from the set \( H(Q, \hat{D}, N) \).
		\EndIf
		\EndFor
		
		\If {the set \( H(Q, \hat{D}, N) \) consists of \( k \) texts or fewer}
		\State Pass the texts in set \( H(Q, \hat{D}, N) \) to Phase II, allowing the LLM to generate the final response to query \( Q \).
		\ElsIf {\( H(Q, \hat{D}, N) \) contains more than \( k \) texts}
		\State Choose the top-\( k \) texts from the set \( H(Q, \hat{D}, N) \) and provide them to Phase II, enabling the LLM to generate the final response to query \( Q \).
		\EndIf
	\end{algorithmic}
\end{algorithm}

Algorithm~\ref{our_alg} shows the pseudocode of our proposed \alg.
For a given query \( Q \), we begin by retrieving the top-\( N \) relevant texts in Phase I of the RAG system, forming the set \( H(Q, \hat{D}, N) \). Each text \( R \) in the set \( H(Q, \hat{D}, N) \) is then split into two chunks, and we calculate its PD, PM, and TS scores as per Eq.~(\ref{pd_score}), Eq.~(\ref{pm_score}), and Eq.~(\ref{ts_score}). If any of the conditions in Eq.~(\ref{pd_condition}), Eq.~(\ref{pm_condition}), or Eq.~(\ref{ts_condition}) are met, \( R \) is identified as poisoned and removed from \( H(Q, \hat{D}, N) \). 
If \( H(Q, \hat{D}, N)\) contains \( k \) or fewer texts after filtering, these are passed to Phase II for the LLM to generate the final response. If \( H(Q, \hat{D}, N) \) still has more than \( k \) texts, only the top-\( k \) are selected, as RAG systems typically limit input to the most relevant \( k \) texts due to LLM token constraints, which avoids truncation of important content~\cite{karpukhin2020dense,lewis2020retrieval,izacard2020leveraging,guu2020retrieval}. If filtering results in an empty set \(H(Q, \hat{D}, N) \), we retrieve a larger set (e.g., top-\( 2N \)) in Phase I to ensure sufficient input.

\subsection{Formal Security Analysis}

In this section, we present the theoretical security analysis of RAGuard. Let output accuracy (OACC) be the fraction of target queries correctly answered after applying RAGuard. Theorem~\ref{theorem_1} shows that if less than half of the poisoned texts bypass filtering, RAGuard guarantees a lower bound on OACC, ensuring reliable performance under poisoning attacks.

\begin{thm}%
	\label{theorem_1}
	Suppose the RAG system retrieves a final set of $k$ texts after applying \alg filtering, and assume the LLM answers correctly whenever more than half of these $k$ texts are benign. Let $\rho$ be the fraction of poisoned texts in the knowledge database, and let $\beta_{\mathrm{PD}}, \beta_{\mathrm{PM}}, \beta_{\mathrm{TS}}$ be the false negative rates of the Perplexity Difference, Perplexity Maximum, and Text Similarity filters. Since a poisoned text must bypass all filters to survive, its survival probability is bounded by
	$
	\beta_{\mathrm{total}} \;\leq\; \beta_{\mathrm{PD}} \cdot \beta_{\mathrm{PM}} \cdot \beta_{\mathrm{TS}}.
	$
	Define the effective poisoned fraction after filtering as $\rho \beta_{\mathrm{total}}$, i.e., the probability that a retrieved text is poisoned and survives all \alg filters. If $\rho \beta_{\mathrm{total}} < \tfrac{1}{2}$, then the output accuracy (OACC) of the defended RAG system satisfies:
	\[
	\mathrm{OACC} \;\geq\; 1 - \exp(-c k),
	\]
	where
	$
	c = \tfrac{1}{3}\,(\tfrac{1}{2} - \rho \beta_{\mathrm{total}})^2 \rho \beta_{\mathrm{total}}.
	$
\end{thm}%

\begin{proof}
   The proof is relegated to Appendix~\ref{sec:appendix_1}.
\end{proof}

%% file: experiments.tex

\section{Experiments} \label{sec:exp}

\subsection{Experimental Setup}

\subsubsection{Datasets}

We evaluate our \alg on five datasets. The Natural Questions (NQ)~\cite{kwiatkowski2019natural} dataset contains 2,681,468 queries and answers, while MS-MARCO~\cite{nguyen2016ms} includes 8,841,823 queries with related texts. HotpotQA~\cite{yang2018hotpotqa} provides 113,000 QA pairs for multi-hop reasoning and interpretability evaluation. We also construct two extended datasets, Extended NQ (ENQ) and Extended MS-MARCO (EMS-MARCO), by rewriting NQ and MS-MARCO texts with an LLM, resulting in 2,686,468 and 8,846,823 entries, respectively.

\subsubsection{Poisoning Attacks} 

To evaluate the robustness of our \alg, we use four attacks by default, including Prompt injection attack~\cite{liu2024formalizing}, General trigger attack~\cite{chaudhari2024phantom}, Jamming attack~\cite{shafran2024machine}, and PoisonedRAG attack~\cite{zou2024poisonedrag}.

\subsubsection{Compared Detection Methods}

By default, we evaluate the performance of \alg against two baseline detection methods: Perplexity (PPL)~\cite{jelinek1980interpolated} and PPL window~\cite{liu2024formalizing}. To provide a more comprehensive analysis, we also compare \alg with a more advanced detection-based method, TrustRAG~\cite{zhou2025trustrag}, as well as three prevention-based methods: paraphrasing~\cite{jain2023baseline}, duplicate text filtering~\cite{weis2004detecting}, and RobustRAG~\cite{xiang2024certifiably}, in Section~\ref{sec:discussion}.

\subsubsection{Evaluation Metrics}

We used four key metrics to evaluate the performance and effectiveness of our defense mechanism. For the following four metrics, higher DACC and OACC, along with lower FPR and FNR, indicate better defense performance.

\myparatight{Detection accuracy (DACC)} DACC was used to measure the precision of our defense in detecting poisoned texts. It is defined as the proportion of correctly classified samples (both positive and negative) out of the total number of samples.

\myparatight{False positive rate (FPR)} FPR measures the rate at which clean/benign texts are incorrectly classified as poisoned.

\myparatight{False negative rate (FNR)} FNR measures the proportion of poisoned texts that the model fails to detect.

\myparatight{Output accuracy (OACC)}OACC represents the proportion of target queries for which the LLM generates correct answers. A higher OACC reflects greater accuracy and effectiveness of the model in producing reliable responses.

\subsubsection{Target Queries and Target Answers}Following~\cite{zou2024poisonedrag}, we use the same 100 target queries and answers for each dataset. 
Note that in line with ~\cite{zou2024poisonedrag}, all 100 target queries are close-ended questions; we do not include any open-ended questions.

\begin{table}[t]
	\centering
	\tiny
	\addtolength{\tabcolsep}{-5.56pt}
	\caption{Detection results of different methods. Larger (\textcolor{blue}{$\uparrow$}) DACC and OACC, and smaller (\textcolor{blue}{$\downarrow$}) FPR and FNR, indicate better detection performance.}
	\label{main_results}%
	\renewcommand\arraystretch{1.1}
	\begin{tabular}{|c|c|cccc|cccc|cccc|}
			\hline
			\multirow{3}{*}{Dataset} & \multirow{3}{*}{Attack} & \multicolumn{12}{c|}{Detection method} \\
			\cline{3-14}          &       & \multicolumn{4}{c|}{PPL}      & \multicolumn{4}{c|}{PPL window} & \multicolumn{4}{c|}{\alg} \\
			\cline{3-14}         
			&       & DACC\textcolor{blue}{$\uparrow$} & FPR\textcolor{blue}{$\downarrow$}   & FNR\textcolor{blue}{$\downarrow$}   &OACC\textcolor{blue}{$\uparrow$} & DACC\textcolor{blue}{$\uparrow$}& FPR\textcolor{blue}{$\downarrow$} & FNR\textcolor{blue}{$\downarrow$}& OACC\textcolor{blue}{$\uparrow$}& DACC\textcolor{blue}{$\uparrow$}& FPR\textcolor{blue}{$\downarrow$} & FNR\textcolor{blue}{$\downarrow$} & OACC\textcolor{blue}{$\uparrow$} \\
			\hline
			\hline
			\multirow{5}{*}{NQ} & No attack &       NA&       0.197&       NA&       1.000&       NA&       0.136&       NA&       1.000&       NA&       0.043&       NA&  1.000\\
			& Prompt injection &       0.649&       0.164&       0.351&       0.695&       0.810&       0.076&       0.190&       0.857&       0.921&       0.057&       0.079&  0.982\\
			& General trigger &       0.638&       0.159&       0.362&       0.786&       0.803&       0.082&       0.197&       0.849&       0.916&       0.054&       0.084&  0.978\\
			& Jamming &       0.756&       0.116&       0.244&       0.806&       0.814&       0.129&       0.186&       0.865&       0.988&       0.098&       0.002&  1.000\\
			& PoisonedRAG &       0.593&       0.697&       0.407&       0.659&       0.708&       0.204&       0.292&       0.807&       0.962&       0.028&       0.038&  0.999\\
			\hline
			\multirow{5}{*}{{\makecell {MS-\\MARCO}}} & No attack &       NA&       0.205&       NA&       0.998&       NA&       0.124&       NA&       1.000&       NA&       0.025&       NA&  1.000\\
			& Prompt injection &       0.615&       0.162&       0.385&       0.713&       0.758&       0.054&       0.242&       0.852&       0.927&       0.048&       0.073&  0.992\\
			& General trigger &       0.622&       0.173&       0.378&       0.726&       0.761&       0.058&       0.239&       0.863&       0.921&       0.029&       0.079&  0.986\\
			& Jamming &       0.734&       0.074&       0.266&       0.831&       0.829&       0.089&       0.168&       0.873&       0.999&       0.034&       0.001&  0.999\\
			& PoisonedRAG &       0.584&       0.396&       0.416&       0.668&       0.715&       0.197&       0.285&       0.819&       0.924&       0.039&       0.076&  0.973\\
			\hline
			\multirow{5}{*}{HotpotQA} & No attack &       NA&       0.217&       NA&       0.990&       NA&       0.085&       NA&       0.996&       NA&       0.063&       NA&  0.998\\
			& Prompt injection &       0.638&       0.199&       0.362&       0.835&       0.831&       0.083&       0.169&       0.859&       0.918&       0.053&       0.082&  0.971\\
			& General trigger &       0.616&       0.135&       0.384&       0.828&       0.817&       0.086&       0.183&       0.867&       0.915&       0.047&       0.085&  0.970\\
			& Jamming &       0.721&       0.167&       0.279&       0.831&       0.832&       0.098&       0.168&       0.872&       0.999&       0.080&       0.001&  0.999\\
			& PoisonedRAG &       0.413&       0.704&       0.587&       0.515&       0.726&       0.193&       0.274&       0.823&       0.941&       0.099&       0.059&  0.988\\
			\hline
			\multirow{5}{*}{ENQ} & No attack &       NA&       0.164&       NA&       0.996&       NA&       0.080&       NA&       1.000&       NA&       0.044&       NA&  1.000\\
			& Prompt injection &       0.654&       0.156&       0.346&       0.702&       0.804&       0.079&       0.196&       0.837&       0.924&       0.038&       0.076&  0.984\\
			& General trigger &       0.647&       0.159&       0.353&       0.698&       0.806&       0.077&       0.194&       0.849&       0.920&       0.037&       0.080&  0.979\\
			& Jamming &       0.753&       0.109&       0.247&       0.806&       0.815&       0.132&       0.185&       0.868&       0.998&       0.066&       0.002&  1.000\\
			& PoisonedRAG &       0.627&       0.448&       0.373&       0.715&       0.711&       0.135&       0.289&       0.816&       0.968&       0.019&       0.032&  0.999\\
			\hline
			\multirow{5}{*}{{\makecell {EMS-\\MARCO}}} & No attack &       NA&       0.137&       NA&       0.998&       NA&       0.049&       NA&       0.999&       NA&       0.023&       NA&  1.000\\
			& Prompt injection &       0.608&       0.114&       0.368&       0.615&       0.764&       0.046&       0.239&       0.858&       0.928&       0.032&       0.073&  0.992\\
			& General trigger &       0.618&       0.125&       0.381&       0.722&       0.756&       0.043&       0.244&       0.857&       0.921&       0.019&       0.079&  0.987\\
			& Jamming &       0.738&       0.056&       0.243&       0.833&       0.835&       0.058&       0.155&       0.882&       0.999&       0.022&       0.001&  0.999\\
			& PoisonedRAG &       0.576&       0.297&       0.423&       0.664&       0.709&       0.142&       0.267&       0.802&       0.924&       0.026&       0.077&  0.978\\
			\hline
		\end{tabular}
		\vspace{-0.15in}
	\end{table}

\subsubsection{Parameter Settings}
Following~\cite{zou2024poisonedrag}, the attacker injects 5 poisoned texts per target query during attacks and used GPT-3.5 as the final evaluation model. 
To ensure efficient evaluation of the model, we predominantly use the GPT-2 model for perplexity calculations and use the dot product method to compute text similarity. The Contriver~\cite{izacard2021unsupervised} model acts as a retrieval module in RAG, responsible for identifying highly relevant texts for each query.
By default, \( k \) is set to 5 following~\cite{zou2024poisonedrag}; in our \alg, we set \( N = 3k \). In \alg, we randomly sample 1000 texts from the knowledge database to form the set \( \mathcal{S} \).
The significance level is set \( \alpha = 2.5\% \).

\begin{figure*}[htbp]
	\center
\subfloat{\includegraphics[width=0.23\textwidth]{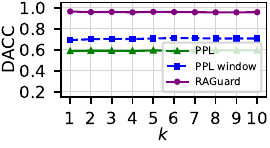}}\hspace{0.005\textwidth}
\subfloat{\includegraphics[width=0.23\textwidth]{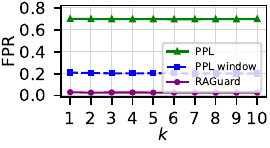}}\hspace{0.005\textwidth}
\subfloat{\includegraphics[width=0.23\textwidth]{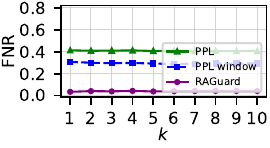}}\hspace{0.005\textwidth}
\subfloat{\includegraphics[width=0.23\textwidth]{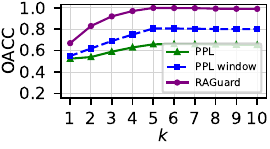}}
	\caption{Impact of $k$ on NQ dataset.}
	\label{fig:Impact of k}
      \vspace{-0.05in}
\end{figure*}

\begin{figure*}[htbp]
	\center
	\subfloat{\includegraphics[width=0.23\textwidth]{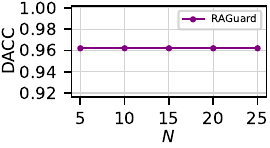}}\hspace{0.005\textwidth}
	\subfloat{\includegraphics[width=0.23\textwidth]{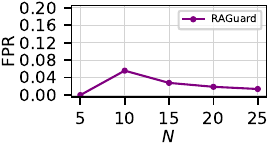}}\hspace{0.005\textwidth}
	\subfloat{\includegraphics[width=0.23\textwidth]{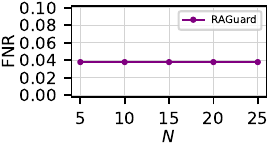}}\hspace{0.005\textwidth}
	\subfloat{\includegraphics[width=0.23\textwidth]{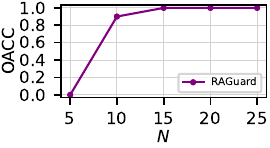}}
	\caption{Impact of $N$ on NQ dataset.}
	\label{fig:Impact of N}
      \vspace{-0.15in}
\end{figure*}

\begin{figure*}[htbp]
	\center
	\subfloat{\includegraphics[width=0.23\textwidth]{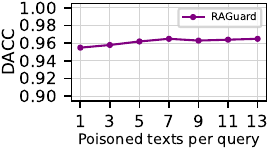}}\hspace{0.005\textwidth}
	\subfloat{\includegraphics[width=0.23\textwidth]{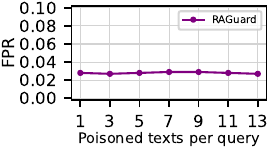}}\hspace{0.005\textwidth}
	\subfloat{\includegraphics[width=0.23\textwidth]{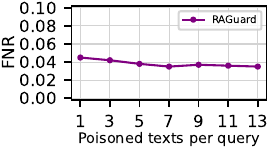}}\hspace{0.005\textwidth}
	\subfloat{\includegraphics[width=0.23\textwidth]{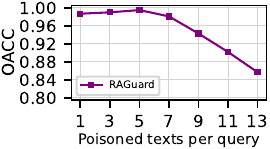}}
	\caption{Impact of poisoned texts per query on NQ dataset.}
	\label{fig:Impact of p}
      \vspace{-0.15in}
\end{figure*}

\begin{figure}[htbp]
    \centering
    \includegraphics[width=0.7\linewidth]{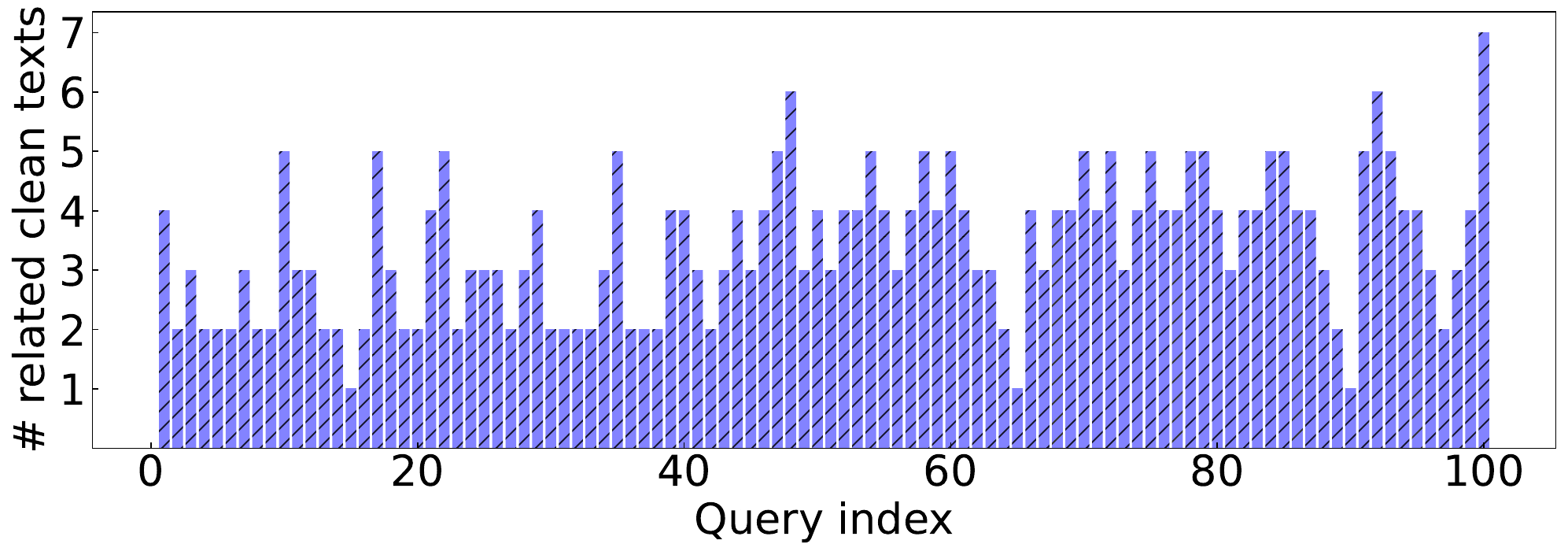}
    \caption{Number of related clean texts of each query.}
    \label{fig:related_query}
  \vspace{-0.15in}
\end{figure}

\subsection{Experimental Results}

\myparatight{\alg achieves satisfactory detection performance}Table~\ref{main_results} presents the defense effects of three methods against four attacks across five datasets. ``NA'' denotes not applicable. 
Key observations from the results include:
First of all, \alg can defend against various attacks, and achieves high DACC and OACC. 
For instance, during the PoisonedRAG attack, \alg achieves DACC scores of 96.2\% on NQ, 92.4\% on MS-MARCO, 94.1\% on HotpotQA, 96.8\% on ENQ, and 92.4\% on EMS-MARCO, with OACC nearing 100\% across all datasets. In addition, \alg also achieves a very low FNR under various attacks on all datasets. Our \alg outperforms the PPL and PPL window methods with about a 30\% and 20\% increase in DACC, respectively, and about a 15\% improvement in OACC over both methods.

Secondly, \alg achieves a low FPR, typically below 6\% across datasets, indicating it rarely removes clean texts by mistake. Although \alg may occasionally delete some clean content, this content is generally irrelevant or of low quality, minimally impacting the LLM’s response accuracy. Consequently, removing detected poisoned texts leads to a higher OACC. In contrast, existing methods (PPL, PPL window) remove more poisoned texts but also mistakenly delete a large amount of clean content, making \alg more effective.

\myparatight{Impact of $k$}This section explores how varying \( k \) impacts the defense effectiveness of different detection methods. 
Note that in all ablation study experiments, we consider only the PoisonedRAG attack by default, as it has proven to be more effective in compromising RAG~\cite{zou2024poisonedrag}.
Fig.~\ref{fig:Impact of k} shows results on the NQ dataset.
The figures reveal that increasing \( k \) has minimal impact on \alg's DACC, FPR, or FNR but does affect OACC, indicating that while \alg's filtering remains consistent, the LLM's accuracy improves with more input texts. In our experiments, setting \( k \geq 5 \) allows the LLM to reliably produce correct answers post-filtering, highlighting the importance of sufficient clean input for accuracy, especially against poisoning attacks.

\myparatight{Impact of $N$}We analyze the impact of \( N \) on \alg's effectiveness. Fig.~\ref{fig:Impact of N} presents the results for the NQ dataset.
We observe that adjusting $N$ has minimal impact on DACC and FNR, which remain stable. However, FPR initially rises and then decreases as $N$ increases, while OACC steadily improves and eventually stabilizes. This suggests that while $N$ has little effect on DACC and FNR, optimizing it can reduce false positives and enhance overall accuracy by increasing the pool of retrieved texts for \alg to process.

\myparatight{Impact of poisoned texts per query}%
We examine how varying the number of poisoned texts per query affects the performance of \alg. Results on NQ dataset are presented in Fig.~\ref{fig:Impact of p}. Note that PPL and PPL window methods are excluded from this analysis, as they consistently fail to detect poisoned texts (see Table~\ref{main_results}).
We observe that as the number of poisoned texts per query increases, the DACC, FPR, and FNR of \alg remain largely stable. However, OACC begins to decline when the number of poisoned texts exceeds 7 per query. This performance drop is attributable to a lack of sufficient clean texts: as shown in Fig.~\ref{fig:related_query}, most queries in the knowledge database contain no more than 5 relevant clean texts. Consequently, once the number of poisoned texts surpasses this threshold, the probability that the retrieval step includes mostly or only poisoned content increases. In our experimental setup, following~\cite{zou2024poisonedrag}, the attacker injects 5 poisoned texts per query into the knowledge database.
Note that injecting such a large number of poisoned texts per query is impractical in real-world scenarios. Knowledge databases are often monitored, and the sudden presence of excessive anomalous entries is likely to raise suspicion. Furthermore, generating a large volume of poisoned texts imposes considerable costs on the attacker, reducing the feasibility and stealth of such attacks~\cite{zou2024poisonedrag}.

\myparatight{Impact of similarity metrics}In Phase I of RAG, knowledge retrieval measures text similarity by comparing the embedding of two texts. Here, we examine how different similarity metrics, such as dot product and cosine similarity, affect the performance of our \alg. The results are shown in Table~\ref{similarity_metrics}. 
We observe that \alg achieves nearly equivalent performance regardless of which similarity metric is applied. Both the dot product metric and the cosine product metric lead to comparable results in terms of text similarity calculations, demonstrating that \alg is robust to the choice of metric in this context.

\myparatight{Impact of different LLMs on computing the perplexity score}In our method, we split the query into two parts and calculate the perplexity score for each separately. This section examines how different LLMs impact perplexity scoring in \alg, with results presented in Table~\ref{perplexity calculation}. We find that GPT-2 performs best, while Bloomz-560M is slightly less effective, and Facebook-1.3B performs the worst. These results suggest a positive correlation between the number of LLM parameters and effectiveness of \alg in perplexity scoring, as models with more parameters, like GPT-2, generally yield better results.

\begin{table}[t]
  \centering
      \addtolength{\tabcolsep}{-2.59pt}
          \tiny
  \caption{Detection results of \alg  with various similarity metrics.}
    \label{similarity_metrics}%
         \renewcommand\arraystretch{1.1}
    \begin{tabular}{|c|c|cccc|cccc|}
    \hline
    \multirow{3}{*}{Dataset} & \multirow{3}{*}{Attack} & \multicolumn{8}{c|}{Similarity} \\
\cline{3-10}          &       & \multicolumn{4}{c|}{Dot Product} & \multicolumn{4}{c|}{Cosine} \\
\cline{3-10}          &       & DACC\textcolor{blue}{$\uparrow$}  & FPR\textcolor{blue}{$\downarrow$}   & FNR\textcolor{blue}{$\downarrow$}   & OACC\textcolor{blue}{$\uparrow$}  & DACC\textcolor{blue}{$\uparrow$}  & FPR\textcolor{blue}{$\downarrow$}   & FNR\textcolor{blue}{$\downarrow$}   & OACC\textcolor{blue}{$\uparrow$} \\
    \hline
    \hline
    \multirow{4}{*}{NQ} & Prompt injection &       0.921 &       0.057&       0.079&       0.988&       0.923&       0.061&       0.077&  0.984\\
& General trigger  &       0.911&       0.050&       0.089&       0.985&       0.916&       0.054&       0.084&  0.978\\
& Jamming &       0.995&       0.096&       0.005&       1.000&       0.998&       0.098&       0.002&  1.000\\
& PoisonedRAG &       0.952&       0.037&       0.046&       0.995&       0.962&       0.028&       0.038&  0.999\\
    \hline
    \multirow{4}{*}{MS-MARCO} & Prompt injection &       0.934&       0.053&       0.066&       0.983&       0.927&       0.048&       0.073&  0.992\\
& General trigger  &       0.924&       0.038&       0.076&       0.993&       0.921&       0.029&       0.079&  0.986\\
& Jamming &       0.993&       0.042&       0.007&       0.992&       0.999&       0.034&       0.001&  0.999\\
& PoisonedRAG &       0.920&       0.037&       0.080&       0.963&       0.924&       0.039&       0.076&  0.973\\
    \hline
    \multirow{4}{*}{HotpotQA} & Prompt injection &       0.921&       0.047&       0.079&       0.976&       0.918&       0.053&       0.082&  0.971\\
& General trigger  &       0.919&       0.048&       0.081&       0.961&       0.915&       0.047&       0.085&  0.970\\
& Jamming &       0.997&       0.078&       0.003&       1.000&       0.999&       0.080&       0.001&  1.000\\
& PoisonedRAG &       0.937&       0.095&       0.063&       0.986&       0.941&       0.099&       0.059&  0.988\\
    \hline
    \end{tabular}%
\end{table}%

\begin{table}[t]
  \centering
          \tiny
  \addtolength{\tabcolsep}{-5.33pt}
  \caption{Detection results of \alg with different LLMs to compute the perplexity score. ``Pinject'' and ``GTrigger'' refer to the ``Prompt injection'' and ``General trigger'', respectively.}
    \label{perplexity calculation}%
     \renewcommand\arraystretch{1.1}
    \begin{tabular}{|c|c|cccc|cccc|cccc|}
    \hline
    \multirow{3}{*}{Dataset} & \multirow{3}{*}{Attack} & \multicolumn{12}{c|}{LLM} \\
\cline{3-14}          &       & \multicolumn{4}{c|}{GPT-2}    & \multicolumn{4}{c|}{Facebook-1.3B} & \multicolumn{4}{c|}{Bloomz-560M} \\
\cline{3-14}          &       & DACC\textcolor{blue}{$\uparrow$}  & FPR\textcolor{blue}{$\downarrow$}   & FNR\textcolor{blue}{$\downarrow$}   & OACC\textcolor{blue}{$\uparrow$}  & DACC\textcolor{blue}{$\uparrow$}  & FPR\textcolor{blue}{$\downarrow$}   & FNR\textcolor{blue}{$\downarrow$}   & OACC\textcolor{blue}{$\uparrow$}  & DACC\textcolor{blue}{$\uparrow$}  & FPR\textcolor{blue}{$\downarrow$}   & FNR\textcolor{blue}{$\downarrow$}   & OACC\textcolor{blue}{$\uparrow$} \\
    \hline
      \hline
    \multirow{4}{*}{NQ} & Pinject &       0.921&       0.057&       0.079&       0.982
&       0.895&       0.038&       0.105&       0.947&       0.917&       0.052&       0.083&  0.975\\
& GTrigger  &       0.916&       0.054&       0.084&       0.978
&       0.883&       0.046&       0.117&       0.929&       0.909&       0.037&       0.091&  0.965\\
& Jamming &       0.988&       0.098&       0.002&       1.000
&       0.979&       0.062&       0.021&       0.985&       0.990&       0.075&       0.010&  0.996\\
& PoisonedRAG &       0.962&       0.028&       0.038&       0.999
&       0.934&       0.025&       0.066&       0.982&       0.951&       0.034&       0.049&  0.988\\
    \hline
    \multirow{4}{*}{{\makecell {MS-\\ MARCO}}} &Pinject &       0.927&       0.048&       0.073&       0.992
&       0.902&       0.031&       0.098&       0.961&       0.922&       0.044&       0.078&  0.983\\
& GTrigger &       0.921&       0.029&       0.079&       0.986
&       0.889&       0.025&       0.111&       0.938&       0.915&       0.036&       0.085&  0.970\\
& Jamming &       0.998&       0.098&       0.002&       1.000
&       0.979&       0.062&       0.021&       0.985&       0.990&       0.075&       0.010&  0.996\\
& PoisonedRAG &       0.924&       0.039&       0.076&       0.973
&       0.897&       0.036&       0.103&       0.951&       0.912&       0.046&       0.088&  0.963\\
    \hline
    \multirow{4}{*}{HotpotQA} &Pinject &       0.918&       0.053&       0.082&       0.971
&       0.894&       0.029&       0.106&       0.945&       0.920&       0.049&       0.080&  0.974\\
& GTrigger &       0.915&       0.047&       0.085&       0.970&       0.891&       0.044&       0.109&       0.942&       0.916&       0.043&       0.084&  0.973\\
& Jamming &       0.999&       0.080&       0.001&       0.999
&       0.974&       0.058&       0.026&       0.984&       0.994&       0.076&       0.006&  0.998\\
& PoisonedRAG &       0.941&       0.099&       0.059&       0.988
&       0.916&       0.064&       0.084&       0.973&       0.933&       0.087&       0.067&  0.979\\
    \hline
    \end{tabular}
 \vspace{-0.1in}
\end{table}

\begin{table}[t]
  \centering
        \tiny
  \caption{OACC of \alg when using different LLMs as the final evaluation model.}
    \label{LLMs_RAG}%
      \renewcommand\arraystretch{1.1}
    \begin{tabular}{|c|c|c|c|c|c|}
    \hline
    \multirow{2}{*}{Dataset} & \multirow{2}{*}{Attack} & \multicolumn{4}{c|}{LLM} \\
\cline{3-6}          &       & GPT-3.5 & GPT-4 & Llama3.1-8B & Kimi \\
    \hline
      \hline
    \multirow{4}{*}{NQ} & Prompt injection &       0.98&       0.99&       0.97&  0.98\\
& General trigger  &       0.97&       0.99&       0.97&  0.98\\
& Jamming &       0.99&       0.99&       0.98&  0.98\\
& PoisonedRAG &       0.99&       1.00&       0.98&  0.99\\
    \hline
    \multirow{4}{*}{{\makecell {MS-MARCO}}} & Prompt injection &       0.99&       0.99&       0.98&  0.99\\
& General trigger  &       0.98&       0.99&       0.97&  0.98\\
& Jamming &       0.99&       1.00&       0.99&  0.99\\
& PoisonedRAG &       0.97&       0.98&       0.96&  0.96\\
    \hline
    \multirow{4}{*}{HotpotQA} & Prompt injection &       0.98&       0.99&       0.98&  0.98\\
& General trigger  &       0.97&       0.98&       0.95&  0.97\\
& Jamming &       0.99&       1.00&       0.98&  0.99\\
& PoisonedRAG &       0.98&       1.00&       0.98&  0.98\\
    \hline
    \end{tabular}
\end{table}

\begin{table}[t]
  \centering
  \tiny
  \caption{Detection performance of TrustRAG.}
  \renewcommand\arraystretch{1.1}
    \begin{tabular}{|c|c|c|c|c|c|}
    \hline
    Dataset & Attack & DACC\textcolor{blue}{$\uparrow$}  & FPR\textcolor{blue}{$\downarrow$}   & FNR\textcolor{blue}{$\downarrow$}   & OACC\textcolor{blue}{$\uparrow$} \\
  \hline
   \hline
    \multirow{5}{*}{NQ} & No attack   &  NA     &  0.205     &  NA     &  0.742 \\
\cline{2-6}          & Prompt injection  &  1.000     &  0.545     &  0.000     &  0.605 \\
\cline{2-6}          & General trigger  &   0.996    &  0.516     &   0.004    & 0.723 \\
\cline{2-6}          & Jamming&   1.000    &  0.563     & 0.000      & 0.646 \\
\cline{2-6}          & PoisonedRAG  &  0.941     & 0.552      &    0.059   & 0.562 \\
  \hline
    \multirow{5}{*}{MS-MARCO} & No attack   &  NA     &  0.472     &   NA    & 0.906 \\
\cline{2-6}          & Prompt injection   &   0.988    &   0.642    & 0.012      & 0.813 \\
\cline{2-6}          &General trigger     & 0.974      &    0.636   &   0.026    &  0.857\\
\cline{2-6}          &Jamming   &   0.965    &  0.674     & 0.350      & 0.824 \\
\cline{2-6}          & PoisonedRAG     & 0.921      &  0.673     &  0.079     &  0.752\\
 \hline
    \multirow{5}{*}{HotpotQA} & No attack  &    NA   & 0.331      &   NA    &  0.865\\
\cline{2-6}          & Prompt injection  &   1.000    & 0.734      & 0.000      &  0.628\\
\cline{2-6}          & General trigger   & 1.000      &   0.665    &    0.000   & 0.613 \\
\cline{2-6}          & Jamming   &   1.000    &0.744      & 0.000      &0.758  \\
\cline{2-6}          & PoisonedRAG  &  0.984     &  0.686     &   0.016    &  0.765\\
 \hline
    \end{tabular}%
  \label{TrustRAG_result}%
          \vspace{-0.1in}
\end{table}%

\begin{table}[t]
  \centering
          \tiny
  \caption{Compare \alg with prevention-based defenses, using PoisonedRAG attack and reporting the OACC.}
      \label{other defense methods}%
      \renewcommand\arraystretch{1.1}
    \begin{tabular}{|c|c|c|c|c|c|}
    \hline
    \multirow{2}{*}{Dataset} & \multicolumn{5}{c|}{Defense} \\
    \cline{2-6}   & No defense   & Paraphrasing & Duplicate  & RobustRAG & \alg \\
   \hline
     \hline
    NQ     &       0.000&       0.175&       0.132& 0.632 &0.999 \\
    \hline
    MS-MARCO    &       0.000&       0.294&       0.206&  0.574 & 0.973  \\
   \hline
 HotpotQA& 0.000& 0.231& 0.167 &  0.530   &0.988\\ \hline
    \end{tabular}%
      \vspace{-0.15in}
\end{table}%

\begin{table*}[h]
    \centering
    \caption{Different variants of \alg, where PoisonedRAG attack is considered.}
          \tiny
    \addtolength{\tabcolsep}{-1.78pt}
    \label{variants}
        \begin{tabular}{|c|cccc|cccc|cccc|cccc|cccc|} \hline
         \multirow{3}{*}{Dataset}&  \multicolumn{20}{c|}{Variant} \\ 
         \cline{2-21}&\multicolumn{4}{c|}{Variant \uppercase\expandafter{\romannumeral1}}&  \multicolumn{4}{c|}{Variant \uppercase\expandafter{\romannumeral2}}& \multicolumn{4}{c|}{Variant \uppercase\expandafter{\romannumeral3}}&\multicolumn{4}{c|}{Variant \uppercase\expandafter{\romannumeral4}} & \multicolumn{4}{c|}{Variant \uppercase\expandafter{\romannumeral5}}\\ 
 \cline{2-21}& DACC\textcolor{blue}{$\uparrow$}& FPR\textcolor{blue}{$\downarrow$}& FNR\textcolor{blue}{$\downarrow$}& OACC\textcolor{blue}{$\uparrow$}& DACC\textcolor{blue}{$\uparrow$}& FPR\textcolor{blue}{$\downarrow$}& FNR\textcolor{blue}{$\downarrow$}& OACC\textcolor{blue}{$\uparrow$}& DACC\textcolor{blue}{$\uparrow$}& FPR\textcolor{blue}{$\downarrow$}& FNR\textcolor{blue}{$\downarrow$}& OACC\textcolor{blue}{$\uparrow$}& DACC\textcolor{blue}{$\uparrow$}& FPR\textcolor{blue}{$\downarrow$}& FNR\textcolor{blue}{$\downarrow$}&OACC\textcolor{blue}{$\uparrow$} & DACC\textcolor{blue}{$\uparrow$}& FPR\textcolor{blue}{$\downarrow$}& FNR\textcolor{blue}{$\downarrow$}&OACC\textcolor{blue}{$\uparrow$} \\ \hline 
   \hline
         NQ&  0.893&  0.043&  0.107&  0.951&  0.695&  0.047&  0.305&  0.702&  0.928& 0.400& 0.073& 0.984& 0.945& 0.143& 0.055&0.986 & 0.810& 0.133& 0.190&0.854\\ \hline 
         MS-MARCO&  0.650&  0.025&  0.350&  0.704&  0.202&  0.135&  0.798&  0.287&  0.941& 0.460& 0.059& 0.990& 0.920& 0.174& 0.080& 0.978& 0.576& 0.135& 0.424&0.642\\ \hline 
         HotpotQA&  0.705&  0.046&  0.295&  0.754&  0.498&  0.117&  0.502&  0.579&  0.932& 0.431& 0.068& 0.986& 0.925& 0.163& 0.075& 0.974& 0.736& 0.169& 0.264&0.810\\ \hline
    \end{tabular}
      \vspace{-0.15in}
\end{table*}

\myparatight{Impact of different LLMs as the final evaluation model on \alg}By default, we use GPT-3.5 as the final evaluation model. Here, we analyze how different LLMs affect the performance of our \alg when used as the final evaluation model. 
The results, presented in Table~\ref{LLMs_RAG}, display the OACC values. 
The results indicate that the choice of LLM has little effect on overall answer accuracy.

%% file: discussion.tex

\section{Discussion} 
\label{sec:discussion}

\myparatight{Compare \alg with advanced detection-based defense}%
In this section, we present a comparison between our proposed \alg and a more recent and sophisticated detection-based method, TrustRAG~\cite{zhou2025trustrag}. 
TrustRAG first clusters the retrieved texts using K-means, leveraging both cosine similarity and ROUGE metrics to capture underlying patterns that may indicate adversarial manipulation. Once the clusters are formed, it performs an internal assessment to pinpoint and eliminate potentially malicious texts.
Table~\ref{TrustRAG_result} shows the detection performance of TrustRAG on the NQ, MS-MARCO, and HotpotQA datasets. Comparing Table~\ref{main_results} with Table~\ref{TrustRAG_result}, we observe that TrustRAG does not deliver satisfactory detection results, as it introduces a high FPR, indicating that many benign texts are mistakenly identified as poisoned.

\myparatight{Compare \alg with prevention-based defenses}%
In this part, we compare our proposed \alg against three state-of-the-art prevention-based defense mechanisms.
For the Paraphrasing method~\cite{jain2023baseline}, the RAG system first employs an LLM to rewrite or reformulate the queries before passing them to the retrieval component. Duplicate text filtering~\cite{weis2004detecting} is a mechanism designed to eliminate redundant or identical entries in the poisoned database by computing hash values for each text and subsequently removing entries with matching hashes. In RobustRAG~\cite{xiang2024certifiably}, each text is first isolated and independently processed to obtain individual responses; these responses are then combined through a voting-based aggregation mechanism to produce an aggregated final output.

Note that prevention-based defense methods inherently do not detect poisoned texts; therefore, metrics specifically related to detection performance, such as DACC, FPR, and FNR, are not applicable to these methods. Instead, we evaluate and compare their effectiveness using OACC values. In particular, for our proposed method \alg, we first identify and remove poisoned texts from the dataset, then compute the OACC based solely on the remaining clean entries. The results under the PoisonedRAG attack are presented in Table~\ref{other defense methods}. In this table, ``No defense'' refers to the baseline RAG system without any defense mechanism, while ``Duplicate'' represents the duplicate text filtering approach. Detection-based methods such as PPL and PPL window are excluded from this comparison, as Table~\ref{main_results} shows they are ineffective in identifying poisoned texts. We also omit other attacks, since prior work~\cite{zou2024poisonedrag} has shown that PoisonedRAG is more effective than these attacks.

As shown in Table~\ref{other defense methods}, prevention-based defenses exhibit limited effectiveness across all evaluated benchmarks, as reflected by their consistently low OACC values on the NQ, MS-MARCO, and HotpotQA datasets. This performance degradation can be attributed to a fundamental limitation of such defenses: they do not perform explicit detection or removal of poisoned texts. As a result, the poisoned texts remain embedded within the system throughout the retrieval process, continuing to influence model responses and undermining the overall robustness of the RAG system.

\myparatight{Different variants of \alg}In this part, we consider different variants of \alg. 
\textit{Variant I} applies only chunk-wise perplexity filtering, relying solely on the conditions defined by Eq.~(\ref{pd_condition}) and Eq.~(\ref{pm_condition}) without incorporating any additional filtering criteria. \textit{Variant II} utilizes only the text similarity filtering mechanism, employing the criterion specified in Eq.~(\ref{ts_condition}) while ignoring perplexity-based scores. \textit{Variant III} combines the perplexity difference score with the text similarity filter, using both Eq.~(\ref{pd_condition}) and Eq.~(\ref{ts_condition}) to identify poisoned texts. \textit{Variant IV} leverages the perplexity maximum score together with the text similarity filtering strategy, applying Eq.~(\ref{pm_condition}) in conjunction with Eq.~(\ref{ts_condition}) to determine which texts to remove. Finally, \textit{Variant V} modifies the chunk-wise perplexity filtering strategy by eliminating the text-splitting step; instead of dividing the input into two shards, it calculates the perplexity over the entire text and then applies the text similarity criterion from Eq.~(\ref{ts_condition}).

It is important to note that our proposed \alg applies a comprehensive filtering strategy that combines three distinct criteria: perplexity difference (Eq.~(\ref{pd_condition})), perplexity maximum (Eq.~(\ref{pm_condition})), and text similarity (Eq.~(\ref{ts_condition})). These components work together to identify and remove potentially poisoned texts. Table~\ref{variants} presents the performance of several ablated variants of \alg under the PoisonedRAG attack, where each variant omits one or more of the filtering conditions. By comparing the results in Table~\ref{variants} with those in Table~\ref{main_results}, we observe a noticeable drop in performance across all variants when compared to the full \alg method. 
This underscores the value of combining all three filters, as each aids accurate detection of poisoned texts.

\begin{table}
    \centering
    \tiny
      \addtolength{\tabcolsep}{-0.5pt}
    \caption{Results of \alg under adaptive attacks.}
    \label{adaptive attacks}
         \renewcommand\arraystretch{1.1}
    \begin{tabular}{|c|cccc|cccc|} \hline  
                 \multirow{3}{*}{Dataset}&\multicolumn{8}{c|}{Adaptive attack}\\
                 \cline{2-9} & \multicolumn{4}{c|}{Adaptive attack \uppercase\expandafter{\romannumeral1}} & \multicolumn{4}{c|}{Adaptive attack \uppercase\expandafter{\romannumeral2}}\\ 
                 \cline{2-9} & DACC\textcolor{blue}{$\uparrow$}& FPR\textcolor{blue}{$\downarrow$}& FNR\textcolor{blue}{$\downarrow$}& OACC\textcolor{blue}{$\uparrow$}& DACC\textcolor{blue}{$\uparrow$}& FPR\textcolor{blue}{$\downarrow$}& FNR\textcolor{blue}{$\downarrow$}&OACC\textcolor{blue}{$\uparrow$} 
    \\ \hline  
  \hline
                 NQ& 0.952& 0.097& 0.048& 0.982& 0.983& 0.053& 0.017&0.996\\ \hline 
 MS-MARCO& 0.933& 0.039& 0.067& 0.974& 0.992& 0.068& 0.008&0.999\\ \hline
 HotpotQA& 0.938& 0.100& 0.062& 0.978& 0.977& 0.094& 0.023&0.990\\ \hline
 ENQ& 0.957& 0.062& 0.043& 0.991& 0.979& 0.040& 0.021&0.993\\ \hline
 EMS-MARCO& 0.923& 0.026& 0.077& 0.970& 0.989& 0.045& 0.011&0.998\\ \hline 
    \end{tabular}
        \vspace{-0.15in}
\end{table}

\myparatight{Adaptive attacks}%
To further evaluate the robustness of \alg against sophisticated adversaries, we design two strong adaptive attacks that specifically target the defense mechanisms used by \alg. These attacks are constructed based on an understanding of how \alg filters poisoned texts, with the goal of modifying malicious inputs so they closely resemble benign ones and thereby avoid detection. 
The details of the two adaptive attacks are as follows:

\textit{Adaptive attack I:} This attack uses GPT-4 to automatically paraphrase poisoned texts. The resulting outputs are semantically malicious but syntactically crafted to resemble benign samples. The specific prompt used to generate these poisoned texts is presented below:

\begin{tcolorbox}[colback=greyBackground, colframe=black!80,
                  arc=1mm, auto outer arc,
                  boxrule=1pt,
                  left=1mm, 
                  right=1mm,
                  top=1mm,  
                  bottom=1mm,
                  ]
\footnotesize
You are a helpful AI assistant, below are a query and wanted answer. 
Please generate texts which let other LLMs respond with the wanted answer when asked the question. \\
\textbf{Query: [query]}  \\
\textbf{Wanted answer: [wanted answer]} \\
\textbf{Texts generated:}  
\end{tcolorbox}
\textit{Adaptive attack II:} In this attack, poisoned texts are manually rewritten by human annotators. The objective is to restructure and rephrase the original content so that it becomes difficult to distinguish from clean texts, both in terms of style and linguistic features

Table~\ref{adaptive attacks} summarizes the defense performance of \alg across five datasets when evaluated against these two adaptive attack strategies. The results indicate that, despite the attacker’s deliberate efforts to evade detection through paraphrasing, \alg remains highly effective.

\myparatight{Computational overhead of \alg}We consider the runtime of various defense methods to measure the computational overhead. 
``RAG only'' refers to the standard RAG system without any defense mechanism applied.
Fig.~\ref{fig:runtime} compares the running time of different methods, showing their computational efficiency. 
\alg adds only a slight computational overhead compared to the standard RAG system.

\begin{figure}[t]
    \centering
    \includegraphics[width=0.9\linewidth]{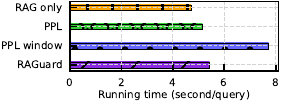}
    \caption{Computational overhead of different methods.}
    \label{fig:runtime}
        \vspace{-0.15in}
\end{figure}

%% file: conclusion.tex

\section{Conclusion} \label{sec:conclusion}

This paper introduces \alg, a defense strategy against poisoning attacks in RAG systems. \alg enhances retrieval by expanding the number of retrieved documents and applying a filtering mechanism that combines perplexity scoring with text similarity to accurately identify and remove malicious content. Extensive experiments demonstrate \alg's effectiveness across various scenarios.

\section*{Acknowledgement}
We thank the anonymous reviewers for their comments.

%% file: appendix.tex

\appendix

\subsection{Proof of Theorem~\ref{theorem_1}} 
\label{sec:appendix_1}

Since at most a $\rho$ fraction of texts in the knowledge database are poisoned, the expected number of poisoned texts in the retrieved set $H(Q,\hat{D},N)$ is bounded by
$
\mathbb{E}[X] \;\leq\; \rho N,
$
where $X$ denotes the number of poisoned candidates prior to filtering.
Now, each text in $H(Q,\hat{D},N)$ is passed through RAGuard's three independent filters. 
Since $
\beta_{\mathrm{total}} \;\leq\; \beta_{\mathrm{PD}}\cdot \beta_{\mathrm{PM}}\cdot \beta_{\mathrm{TS}}
$, thus the expected number of poisoned texts after filtering is
$
\mathbb{E}[Y] \;\leq\; \rho N \beta_{\mathrm{total}},
$ where $Y$ is the number of poisoned texts after filtering.
The final top-$k$ texts are selected from this filtered set. Let $Z$ denote the number of poisoned texts among the top-$k$. Since each text survives independently with probability at most $\rho \beta_{\mathrm{total}}$, we can model $Z$ as being stochastically dominated by a binomial random variable
$
Z \;\preceq\; \text{Binomial}(k, \rho \beta_{\mathrm{total}}).
$
Hence the expectation satisfies
$
\mu = \mathbb{E}[Z] \;\leq\; k \rho \beta_{\mathrm{total}}.
$


Define
$
\delta = \frac{k/2 - \mu}{\mu} > 0.
$
By the multiplicative Chernoff bound, if $Z \sim \text{Binomial}(k,p)$ with mean $\mu = kp$, then for any $\delta > 0$,
$
\Pr[Z \geq (1+\delta)\mu] \leq \exp\!\left(-\frac{\delta^2 \mu}{3}\right).
$
Applying this with $p = \rho \beta_{\mathrm{total}}$, we obtain
$
\Pr[Z \geq k/2] \leq \exp\!\left(-\frac{\delta^2 \mu}{3}\right).
$
Substituting $\mu \leq k \rho \beta_{\mathrm{total}}$ yields
$
\Pr[Z \geq k/2] \;\leq\; \exp(-c k),
$
where 
$
c = \frac{1}{3}(\tfrac{1}{2} - \rho \beta_{\mathrm{total}})^2 \rho \beta_{\mathrm{total}}.
$
Finally, by assumption, the LLM outputs the correct answer whenever $Z < k/2$. Therefore the probability of producing the correct output is
$
\Pr[Z < k/2] = 1 - \Pr[Z \geq k/2] \;\geq\; 1 - \exp(-c k).
$
Since OACC measures the fraction of queries for which the system outputs the correct answer, we conclude that
$
\mathrm{OACC} \;\geq\; 1 - \exp(-c k).
$